\documentclass[11pt]{amsart}
\usepackage{amsmath}
\usepackage{amssymb}
\usepackage{amsfonts}
\usepackage{latexsym}
\usepackage{fancyhdr}
\usepackage{graphics}
\newcommand{\norm}[1]{ \parallel #1 \parallel}
\newcommand{\mb}{\mathbb}
\newcommand{\mc}{\mathcal}
\newcommand{\eul}{\mathfrak}
\newcommand{\A}{\eul A}
\newcommand{\Ao}{{\eul A}_{\scriptscriptstyle 0}}
\newcommand{\Bo}{{\eul B}_{\scriptscriptstyle 0}}
\newcommand{\ze}{{\scriptscriptstyle 0}}
\newcommand{\M}{\eul M}
\newcommand{\D}{\mc D}
\newcommand{\HH}{\mc H}

\newcommand{\id}{{\mb I}}

\newcommand{\FF}{\mathfrak F}

\newcommand{\GK}{{\eul K}}

\newcommand{\X}{{\mathfrak X}}

\newcommand{\vp}{\varphi}
\newcommand{\vna}{von Neumann algebra}
\newtheorem{defn}{Definition}[section]
\newtheorem{prop}[defn]{Proposition}
\newtheorem{thm}[defn]{Theorem}
\newtheorem{lemma}[defn]{Lemma}

\newtheorem{example}[defn]{Example}
\newtheorem{rem}[defn]{Remark}
\def\x{\relax\ifmmode {\mbox{*}}\else*\fi}
\newcommand{\beex}{\begin{example}$\!\!${\bf }$\;$\rm }
\newcommand{\enex}{ \end{example}}
\newcommand{\berem}{\begin{rem}$\!\!${\bf }$\;$\rm }
\newcommand{\enrem}{ \end{rem}}
\newcommand{\aff}{\,\eta\,}
\newcommand{\BH}{\mc{B}(\mathcal{H})}
\newcommand{\frm}[2]{\langle {#1},{#2}\rangle}
\begin{document}
\title{Quasi *-algebras of measurable operators}%
\author{Fabio Bagarello}%
\address{Dipartimento di Metodi e Modelli Matematici, Universit\`a
di Palermo, I-90128 Palermo (Italy)}
\email{bagarell@unipa.it}%
\author{Camillo Trapani}%
\address{Dipartimento di Matematica ed Applicazioni, Universit\`a
di Palermo, I-90123 Palermo (Italy)} \email{trapani@unipa.it}
\author{Salvatore Triolo}%
\address{Dipartimento di Matematica ed Applicazioni, Universit\`a
di Palermo, I-90123 Palermo (Italy)} \email{salvo@math.unipa.it}

\begin{abstract} Non-commutative $L^p$-spaces are shown to constitute examples of a class
of Banach quasi *-algebras called CQ*-algebras. For $p\geq 2$ they
are also proved to possess a {\em sufficient} family of bounded
positive sesquilinear forms satisfying certain invariance
properties.  CQ *-algebras of measurable operators over a finite
von Neumann algebra are also constructed and it is  proven that
any abstract CQ*-algebra $(\X,\Ao)$ possessing a sufficient family
of bounded positive tracial sesquilinear forms  can be represented
as a CQ*-algebra of this type.
\end{abstract}
\maketitle

\section{Introduction and preliminaries}

A {\em quasi *-algebra} is a couple $(\X, \A_{\scriptscriptstyle
0})$, where $\X$ is a vector space with involution $^*$,
$\A_{\scriptscriptstyle 0}$ is a *-algebra and a vector subspace
of $\X$ and $\X$ is an $\A_{\scriptscriptstyle 0}$-bimodule whose
module operations and involution extend those of
$\A_{\scriptscriptstyle 0}$. Quasi *-algebras were introduced by
Lassner \cite{lass3, lass4, schbook} to provide an appropriate
mathematical framework where discussing certain quantum physical
systems for which the usual algebraic approach made in terms of
C*-algebras revealed to be insufficient. In these applications
they usually arise by taking the completion of the C*-algebra of
observables in a weaker topology satisfying certain physical
requirements. The case where this weaker topology is a norm
topology has been considered in a series of previous papers
\cite{cquno}-\cite{standard}, where CQ*-algebras were introduced:
a CQ*-algebra is, indeed, a quasi *-algebra $(\X,
\A_{\scriptscriptstyle 0})$ where $\X$ is a Banach space with
respect to a norm $\| \ \|$ possessing an isometric involution and
$\Ao$ is a C*-algebra with respect to a norm $\| \cdot \|_\ze$,
which is dense in $\X[\| \cdot \|]$.

Since any C*-algebra $\Ao$ has a faithful *-representation $\pi$,
it is natural to pose the question if this completion also can be
realized as a quasi *-algebra of operators affiliated to
$\pi(\Ao)''$. The Segal-Nelson theory \cite{segal, nelson} of
non-commutative integration provides a number of mathematical
tools for dealing with this problem.

The paper is organized as follows. In Section 2 we consider
non-commutative $L^p$-spaces constructed starting from a von
Neumann algebra $\M$ and a normal, semifinite, faithful trace
$\tau$ as Banach quasi *-algebras. In particular if $\vp$ is
finite, then it is shown that $(L^p(\vp), \M)$ is a CQ*-algebra.
If $p \geq 2$, they even possess a {\em sufficient} family of
positive sesquilinear forms enjoying certain {\em invariance}
properties.

In Section 3, starting from a family $\FF$ of normal, finite
traces on a von Neumann algebra $\M$, we prove that the completion
of $\M$ with respect to a norm defined in natural way by the
family $\FF$ is indeed a CQ*-algebra consisting of measurable
operators, in Segal's sense, and therefore affiliated with $\M$.

Finally, in Section 4, we prove that any CQ*-algebra $(\X,\Ao)$
possessing a sufficient family of bounded positive tracial
sesquilinear forms  can be continuosly embedded into the
CQ*-algebra of measurable operators constructed in Section 3.

\bigskip
In order to keep the paper sufficiently self-contained, we collect
below some preliminary definitions and propositions that
will be used in what follows.

\bigskip Let $(\X,\Ao)$ be a quasi *-algebra. The {\it unit} of $(\X, \A_{\scriptscriptstyle 0})$ is an
element $e\in \A_{\scriptscriptstyle 0}$ such that $xe=ex=x$, for every $x \in \X$.
 A quasi *-algebra $(\X, \A_{\scriptscriptstyle 0})$ is said to
be {\it locally convex} if $\X$ is endowed with a topology $\tau$
which makes of $\X$ a locally convex space and such that the
involution $a \mapsto a^*$ and the multiplications $a \mapsto ab$,
$a \mapsto ba$, $b \in \A_{\scriptscriptstyle 0}$, are continuous. If $\tau$ is a norm
topology and the involution is isometric with respect to the norm,
we say that $(\X, \A_{\scriptscriptstyle 0})$ is a {\it normed quasi *-algebra}   and,
if it is complete, we say it is a {\it Banach quasi*-algebra}.

\begin{defn}\label{cqetc} Let $(\X, \A_{\scriptscriptstyle 0})$ be a Banach quasi *-algebra
with norm $\| \cdot \|$ and involution$\,^*$. Assume that  a
second norm $\| \cdot \|_\ze$ is defined on
$\A_{\scriptscriptstyle 0}$, satisfying the following conditions:
\begin{itemize}
\item[(a.1)]$\|a^* a\|_\ze= \|a\|^2_\ze,\quad \forall a \in
\A_{\scriptscriptstyle 0}$;
\item[(a.2)]$\|a\|\leq \|a\|_\ze, \quad \forall a \in \Ao$;
\item[(a.3)]$\|ax\| \leq \|a\|_\ze \|x\|, \quad \forall a \in
\A_{\scriptscriptstyle 0}, x \in \X$;
\item[(a.4)]$\Ao[\| \cdot \|_\ze]$ is complete.
\end{itemize}
Then we say that $(\X, \A_{\scriptscriptstyle 0})$ is a
CQ*-algebra.
\end{defn}

\berem (1) If $\Ao[\| \cdot \|_\ze]$ is not complete, we say that
$(\X, \A_{\scriptscriptstyle 0})$ is a pre CQ*-algebra.

\noindent(2) In previous papers the name CQ*-algebra was given to
a more complicated structure where two different involutions were
considered on $\Ao$. When these involutions coincide, we spoke of
a {\em proper} CQ *-algebra. In this paper only this case will be
considered and so we systematically omit the term {\em
proper}.\enrem

\bigskip
The following basic definitions and results on non-commutative
measure theory are also needed in what follows.

 Let $\M$ be a von Neumann
algebra and $\vp$ a normal faithful semifinite trace defined on
$\M_+$.

Put $${\mc J}=\{ X \in \M: \vp(|X|)<\infty \}.$$ ${\mc J}$ is a
*-ideal of $\M$.

We denote with ${\rm Proj}{(\M)}$, the lattice of projections of
$\M$.
\begin{defn}A vector subspace $\D$ of $\HH$ is said to be strongly
dense ( resp., strongly $\vp$-dense) if
\begin{itemize}
\item $U'\D \subset \D$ for any unitary $U'$ in $\M'$
\item there exists a sequence $P_n \in {\rm Proj}{(\M)}$: $P_n \HH
\subset \D$ , $P_n^\perp \downarrow 0$ and ${(P_n^\perp)}$ is a
finite projection (resp., $\vp{(P_n^\perp)}<\infty$).
\end{itemize}
\label{meassegal}
\end{defn}

Clearly, every strongly $\vp$-dense domain is strongly dense.

Throughout this paper, when we say that an operator $T$ is
affiliated with a von Neumann algebra, written $T \aff \M$, we always mean that $T$ is closed, densely defined and $TU\supseteq UT$ for every unitary operator $U \in \M'$.
\begin{defn}
An operator $T \aff \M$ is called
\begin{itemize}
\item measurable (with respect to
$\M$) if its domain $D(T)$ is strongly dense;
\item $\vp$-measurable if its domain
$D(T)$ is strongly $\vp$-dense. \label{meaurability}
\end{itemize}
\end{defn}
>From the definition itself it follows that, if $T$ is
$\vp$-measurable, then there exists $P \in {\rm Proj}{(\M)}$ such
that $TP$ is bounded and $\vp(P^\perp)<\infty$. \label{main}

We remind that any operator affiliated with a finite von Neumann
algebra is measurable \cite[Cor. 4.1]{segal} but it is not
necessarily $\vp$-measurable.

\section{Non-commutative $L^p$-spaces as CQ *-algebras}

In this Section we will discuss the structure of the
non-commutative $L^p$-spaces as quasi *-algebras. We begin with
recalling the basic definitions.

Let $\M$ be a von Neumann algebra and $\vp$ a normal faithful
semifinite trace defined on $\M_+$. For each $p\geq 1$, let
$${\mc J_p}=\{ X \in \M: \vp(|X|^p)<\infty \}.$$ Then ${\mc J_p}$ is a
*-ideal of $\M$. Following \cite{nelson}, we denote with
$L^p(\vp)$ the Banach space completion of ${\mc J_p}$ with respect
to the norm
$$\|X\|_p := \vp(|X|^p)^{1/p}, \quad X \in {\mc J_p}.$$
One usually defines $L^\infty(\vp) = \M$. Thus, if $\vp$ is a
finite trace, then $L^\infty(\vp) \subset L^p(\vp)$ for every
$p\geq 1$. As shown in \cite{nelson}, if $X \in L^p(\vp)$, then
$X$ is a measurable operator.
\begin{prop} Let $\M$ be a von Neumann
algebra and $\vp$ a normal faithful semifinite trace on $\M_+$.
Then $(L^p(\vp), L^\infty(\vp)\cap L^p(\vp))$ is a Banach quasi
*-algebra.

If $\vp$ is a finite trace and $\vp(\id)=  1$, then $(L^p(\vp),
L^\infty(\vp))$ is a CQ*-algebra.
\end{prop}
\begin{proof} Indeed, it is easily seen that the norms $\| \cdot \|_\infty$ of $L^\infty(\vp)\cap
L^p(\vp)$ and $\| \cdot \|_p$ on $L^p(\vp)$ satisfy the conditions
(a.1)-(a.2) of Definition \ref{cqetc}. Moreover, if $\vp$ is
finite, then $L^\infty(\vp) \subset L^p(\vp)$ and thus $(L^p(\vp),
L^\infty(\vp))$ is a CQ*-algebra.
\end{proof}
\berem Of course the condition $\vp(\id)=  1$ can be easily
removed by rescaling the trace. \enrem
 \begin{defn}Let $(\X,\Ao)$ be a Banach quasi *-algebra.
We denote with $\mc S(\X) $ the set of all sesquilinear forms
$\Omega$ on $\X \times \X$ with the following properties
\begin{itemize}
\item[(i)]$ \Omega(x,x) \geq 0 \hspace{3mm} \forall x \in \X$
\item[(ii)] $\Omega(xa,b) = \Omega (a,x^*b)  \hspace{3mm} \forall x \in
\X,\hspace{2mm} \forall a,b \in \Ao$
\item[(iii)]$ |\Omega(x,y)| \leq \norm{x}
\norm{y} \hspace{3mm} \forall x,y\in \X$ .\end{itemize}
  A subfamily ${\mc A}$ of $\mc S(\X) $ is called {\em sufficient}  if $x \in
\X$,
 $\Omega(x,x)=0 $, for every $\Omega \in {\mc A}$,
implies $x=0$. \label{SLP}
\end{defn}

If $(\X,\Ao)$ is a Banach quasi *-algebra, then the Banach dual
space $\X^\sharp$ of $\X$ can be made into a Banach $\Ao$-bimodule
with norm
$$
\|f\|^\sharp = \sup_{\|x\|\leq 1} |\frm{x}{f}|, \quad f \in
\X^\sharp,
$$ by defining, for $f \in \X^\sharp$, $a \in \Ao$, the module operations in the following way:
$$ \frm{x}{f\circ a}:= \frm{ax}{f} , \quad x \in \X$$
$$ \frm{x}{a\circ f}:= \frm{xa}{f}, \quad x \in \X.$$
As usual, an involution $f \mapsto f^*$ can be defined on $\X^\sharp$ by $\frm{x}{f^*}= \overline{\frm{x^*}{f}}$, $x \in \X$.\\
With these notations we can easily prove the following (see, also
\cite{tratri}):
\begin{prop}\label{repr} $(\X,\Ao)$ be a  Banach quasi *-algebra and $\Omega$ a positive sesquilinear form on $\X \times \X$. The following
statements are equivalent:
\begin{itemize}
\item[(i)] $\Omega \in \mc S(\X) $;
\item[(ii)] there exists a bounded conjugate linear operator $T:\X
\to \X^\sharp$ with the properties
\begin{itemize}
\item[(ii.1)] $\frm{x}{Tx}\, \geq 0, \quad \forall x \in \X;$
\item[(ii.2)] $T(ax)= (Tx)\circ a^* , \quad \forall a \in \Ao, \, x
\in \X;$
\item[(ii.3)] $\|T\|_{\scriptscriptstyle{{\mc B}(\X,\X^\sharp)}}\leq
1;$
\item[(ii.4)] $\Omega(x,y)= \frm{x}{Ty}, \quad \forall x,y \in \X.$
\end{itemize}
\end{itemize}
\end{prop}

 We will now focus our attention on the question
as to whether for the Banach quasi *-algebra
$(L^p(\vp),L^\infty(\vp)\cap L^p(\vp))$, the family ${\mc
S}(L^p(\vp))$, that we are going to describe, is or is not
sufficient.

Before going forth, we remind that many of the familiar results of
the ordinary theory of $L^p$-space hold in the very same form for
the non-commutative $L^p$-spaces. This is the case, for instance,
of H\"older's inequality and also of the statement that
characterizes the dual of $L^p$: the form defining the duality is
the extension of $\vp$ (this extension will be denoted with the
same symbol) to products of the type $XY$ with $X \in L^p(\vp)$,
$Y \in L^{p'}(\vp)$ with $p^{-1}+{p'}^{-1}=1$ and one has
$(L^p(\vp))^\sharp \simeq L^{p'}(\vp)$.

In order to study ${\mc S}(L^p(\vp))$, we introduce, for $p\geq
2$, the following notation
$${\mathcal{B}_+^p}=\{X\in L^{p/(p-2)}(\vp),\, X\geq0\, , \, \, \norm{X}_{p/(p-2)}\leq 1\} $$
where $ p/{(p-2)}=\infty\,$ if $p=2.$

For each $W \in {\mc B}_+^p$, we consider the right multiplication operator
$$ R_{\scriptscriptstyle W}:L^p(\vp) \to L^{\frac{p}{p-1}}(\vp); \hspace{8mm} R_{\scriptscriptstyle W} X= XW, \quad X \in L^p(\vp).$$
Since $L^\infty(\vp)\cap L^p(\vp)= {\mc J}_p$, we use, for
shortness, the latter notation.

\begin{lemma}\label{2.4} The following statements hold.
\begin{itemize}
\item[(i)] Let $p\geq2$. For every
$W\in {\mathcal{B}_+^p}$, the sesquilinear form
$\Omega(X,Y)=\vp[X(R_{\scriptscriptstyle W}Y){^\ast}] $ is an
element of $\mathcal{S}(L^p(\vp))$
\item[(ii)]If $\vp$ is finite, then for each $\Omega \in {\mc S}(L^p(\vp))$, there exists $W \in {\mc B}_+^p$ such that
$$\Omega(X,Y)=\vp[X(R_{\scriptscriptstyle W}Y){^\ast}],\quad \forall X,Y \in L^p(\vp).$$
\end{itemize}
\end{lemma}
\begin{proof}(i):
We check that the sesquilinear
$\Omega(X,Y)=\vp[X(R_{\scriptscriptstyle W}Y){^\ast}]$, $X,Y \in
L^p(\vp)$ satisfies the conditions (i),(ii),(iii) of Definition
\ref{SLP}.

\noindent For every $X\in L^p(\vp)$ we have $$\Omega(X,X)=\vp
[X(R_{\scriptscriptstyle W}X)^\ast]=\vp [X(XW)^{\ast}]=\vp
[(XW)^{\ast}X]=\vp [W|X|^2]\geq 0.$$ For every $X \in L^p(\vp)$,
$A ,B\in {\mc J}_p$, we get
$$\Omega(XA,B)=\vp (XA(BW)^{\ast})=\vp
(WB^{\ast}XA)=\vp(A(X^{\ast}BW)^{\ast})=\Omega(A,X^
 {\ast}B).$$
Finally, for every $ X,Y\in L_p(\vp)$,
$$|\Omega(X,Y)|\leq\norm{X}_p\norm{Y}_{p}\norm{W}_{p/p-2}\leq
\norm{X}_p\norm{Y}_p.$$ (ii) Let $\Omega \in {\mc S}(L^p(\vp))$.
Let $T: L^p(\vp) \to L^{p'}(\vp)$ be the operator which represents
$\Omega$ in the sense of Proposition \ref{repr}. The finiteness of
$\vp$ implies that  ${\mc J_p}=\M$; thus we can put $W= T({\mb
I})$. It is easy to check that $R_{\scriptscriptstyle W}=T$. This
concludes the proof.
\end{proof}

\begin{prop}If $p\geq 2$, ${\mc S}(L^p(\vp))$ is sufficient.
\end{prop}
\begin{proof}
Let $X\in L^p(\vp)$ be such that $ \Omega(X,X)=0 $\, for every
$\Omega$ $\in \mathcal{S}(L_p(\vp))$. By the previous lemma, since
$|X|^{p-2}\in L^{\tfrac{p}{p-2}}(\vp)$, the right multiplication
operator $R_{\scriptscriptstyle W}$ with
$W={\tfrac{|X|^{p-2}}{\alpha}}$,  $ \alpha \in\mathbb{R}$
satisfying $\norm{\tfrac{|X|^{p-2}}{\alpha}}_{p/p-2}\leq 1,$
represents a sesquilinear form $\Omega\in \mathcal{S}(L_p(\vp))$.
By the assumption, $ \Omega(X,X)=0$. We then have
$$\Omega(X,X)=\vp
[X(R_{\scriptscriptstyle W} X)^{\ast}]=\tfrac{\vp
[X(X|X|^{p-2})^{\ast}]}{\alpha}=\tfrac{\vp
[(X|X|^{p-2})^{\ast}X]}{\alpha}=\tfrac{\vp[|X|^p]}{\alpha}=0
\Rightarrow X=0,$$ by the faithfulness of $\vp$.
\end{proof}
%%%%%%%

\section{CQ*-algebras over finite von Neumann algebras}
Let $\M$ be a von Neumann algebra and $\FF=\{\vp_\alpha;\, \alpha
\in {\mc I}\}$ be a family of normal, {\it finite} traces on $\M$.
As usual, we say that the family $\FF$ is {\it sufficient} if for
$X \in \M,\, X \geq 0$ and $\vp_\alpha(X)=0$ for every $\alpha \in
{\mc I}$, then $X=0$ (clearly, if $\FF=\{\vp\}$, then $\FF$ is
sufficient if, and only if, $\vp$ is faithful). In this case, $\M$
is a finite von Neumann algebra \cite[ch.7]{stratila}. We assume,
in addition, that the following condition (P) is satisfied:
$$ \mbox{(P)} \hspace{1cm}
\vp_\alpha ({\mb I})\leq 1, \quad \forall \alpha \in {\mc I}.$$
Then we define
$$ \|X\|_{p,{\mc I}}= \sup_{\alpha \in {\mc I}}\|X\|_{p,\vp_\alpha}= \sup_{\alpha \in {\mc I}} \vp_\alpha(|X|^p)^{1/p}.$$
Since $\FF$ is sufficient, $\| \cdot \|_{p,{\mc I}}$ is a norm on
$\M$.

In the sequel we will need the following Lemmas whose simple
proofs will be omitted.

\begin{lemma}\label{thelemma}
Let $\mathfrak{M}$ be a von Neumann algebra in Hilbert space
$\HH$, $\{P_\alpha\}_{\alpha\in\mathcal{I}}$ a
family of projections of $\mathfrak{M}$ with
$$\bigvee_{\alpha\in\mathcal{I}}P_\alpha=\overline{P}.$$
If $A\in\mathfrak{M}$ and $ AP_\alpha=0 $ for every
$\alpha\in\mathcal{I}$, then $A\overline{P}=0.$
\end{lemma}

\begin{lemma}\label{lemma2}

Let $\FF=\{\vp_\alpha\}_{\alpha\in\mathcal{I}}$ be a sufficient
family of normal, finite traces on the von Neumann algebra
$\mathfrak{M}$ and let $P_\alpha$ be the support of $\vp_\alpha.$
Then, $\vee P_\alpha=\mathbb{I}$, where $\mathbb{I}$ denotes the
identity of $\mathfrak{M}.$
\end{lemma}

\medskip
It is well-known that the support of each $\vp_\alpha$ enjoy the
following properties
\begin{itemize}
  \item[(i)] $P_\alpha\in {\mc Z}(\M)$, the center of $\M$, for
  each $\alpha \in I$;
  \item[(ii)] $\vp_\alpha (X) = \vp_\alpha (XP_\alpha)$, for
  each $\alpha \in
  I$.
\end{itemize}

>From the two preceding lemmas it follows that, if the $P_\alpha$'s
are as in Lemma \ref{lemma2}, then
$$AP_\alpha=0 \quad \forall\alpha\in\mathcal{I} \,\Rightarrow A=0.$$

 If Condition (P) is fulfilled, then
$$\norm{X}_{p,\mathcal{I}}= \sup_{\alpha\in\mathcal{I}}\norm{XP_\alpha}_{p,\alpha} \quad \quad \forall X\in\mathfrak{M} $$
Clearly, the sufficiency of the family of traces and Condition (P)
imply that $\norm{\cdot}_{p,\mathcal{I}}$ is a norm
$\mathfrak{M}.$

\begin{prop} \label{nota}Let ${\mathfrak{M}}(p,{\mc I})$ denote the Banach space completion
of $\mathfrak{M}$ with respect to the norm
$\norm{\cdot}_{p,\mathcal{I}}$. Then $({\mathfrak{M}}(p,{\mc
I})[\norm{\cdot}_{p,\mathcal{I}}],
\mathfrak{M}[\norm{\cdot}_{\scriptscriptstyle\BH}] )$ is a
CQ*-algebra.
\end{prop}

\begin{proof}
Indeed, we have
\begin{equation}\label{aaa}\norm{X{^\ast}}_{p,\mathcal{I}}=\sup_{\alpha\in\mathcal{I}}\norm{X^{\ast}P_\alpha}_{p,\alpha}=
\sup_{\alpha\in\mathcal{I}}\norm{(XP_\alpha)^{\ast}}_{p,\alpha}=\norm{X}_{p,\mathcal{I}},
\quad \forall X\in\mathfrak{M}.\end{equation} Furthermore, for
every $X,Y\in\mathfrak{M}$,
\begin{equation}\label{bbb}\norm{XY}_{p,\mathcal{I}}=\sup_{\alpha\in\mathcal{I}}\norm{XYP_\alpha}_{p,\alpha}\leq
\norm{X}_{\scriptscriptstyle\BH}\sup_{\alpha\in\mathcal{I}}
\norm{YP_\alpha}_{p,\alpha}=\norm{X}_{\scriptscriptstyle\BH}\norm{Y}_{p,\mathcal{I}}.\end{equation}
Finally, condition (P) implies that
$$\norm{X}_{p,\mathcal{I}} \leq
\norm{X}_{\scriptscriptstyle\BH}, \quad \forall
X\in\mathfrak{M}.$$ From \eqref{aaa} and \eqref{bbb} it follows
that ${\mathfrak{M}}(p,{\mc I})$ is a Banach quasi *-algebra. It
is clear that $\norm{\ }_{\scriptscriptstyle\BH}$ satisfies the
conditions (a.1)-(a.3) of Section 1. Therefore
$({\mathfrak{M}}(p,{\mc I}), \mathfrak{M})$ is a CQ *-algebra.
\end{proof}
The next step consists in investigating the Banach space
${\mathfrak{M}}(p,{\mc I})[\norm{\cdot}_{p,\mathcal{I}}]$. In
particular we are interested in the question as to whether
${\mathfrak{M}}(p,{\mc I})[\norm{\cdot}_{p,\mathcal{I}}]$ can be
identified with a space of operators affiliated with $\M$. For
shortness, whenever no ambiguity can arise, we write $\M_p$
instead of ${\mathfrak{M}}(p,{\mc I})$

\medskip
Let $\FF =\{\vp_\alpha\}_{\alpha\in\mathcal{I}}$ be a sufficient
family of normal, finite traces on the von Neumann algebra
$\mathfrak{M}$ satisfying Condition (P). The traces $\vp_\alpha$
are not necessarily faithful. Put $\M_\alpha = \M P_\alpha$,
where, as before, $P_\alpha$ denotes the support of $\vp_\alpha$.
Each $\M_\alpha$ is a \vna\ and $\vp_\alpha$ is faithful in $\M
P_\alpha$ \cite[Proposition V. 2.10]{takesaki}.

More precisely,
$$\M_\alpha:=\M P_\alpha = \{ Z=XP_\alpha, \mbox{ for some } X \in
\M\}.$$

The positive cone $\M_\alpha^+$ of $\M_\alpha$ equals the set

$$\{ Z=XP_\alpha, \mbox{ for some } X \in
\M^+\}.$$ For $Z = XP_\alpha \in \M_\alpha^+$, we put:
$$ \sigma_\alpha (Z):= \vp_\alpha (XP_\alpha).$$
The definition of $\sigma_\alpha (Z)$ does not depend on the
particular choice of $X$. Each $\sigma_\alpha$ is a normal,
finite, faithful trace on $\M_\alpha$. It is then possible to
consider the spaces $L^p(\M_\alpha, \sigma_\alpha)$, $p\geq 1$, in
the usual way. The norm of $L^p(\M_\alpha, \sigma_\alpha)$ is
indicated as $\| \cdot \|_{p,\alpha}$.

\bigskip
Let now $(X_k)$ be a Cauchy sequence in $\M[\| \cdot \|_{p,{\mc
I}}]$. For each $\alpha \in {\mc I}$, we put $Z_k^{(\alpha)}= X_k
P_\alpha$.
 Then,
for each $\alpha \in {\mc I}$, $(Z_k^{(\alpha)})$ is a Cauchy
sequence in $\M_\alpha[\| \cdot \|_{p,\alpha}]$. Indeed, since
$|Z_k^{(\alpha)} - Z_h^{(\alpha)}|^p =|X_k -X_h|^pP_\alpha$,
\begin{eqnarray*}
\|Z_k^{(\alpha)} - Z_h^{(\alpha)}\|_{p,\alpha} &=&
\sigma_\alpha(|Z_k^{(\alpha)} - Z_h^{(\alpha)}|^p)^{1/p} \\
&=& \vp_\alpha(|X_k -X_h|^pP_\alpha)^{1/p} \\ &=& \vp_\alpha(|X_k
-X_h|^p)^{1/p} \to 0.
\end{eqnarray*}
Therefore, for each $\alpha \in {\mc I}$, there exists an operator
$Z^{(\alpha)}\in L^p(\M_\alpha, \sigma_\alpha)$ such that:
$$ Z^{(\alpha)}=\| \cdot \|_{p,\alpha}-\lim_{k \to \infty}
Z_k^{(\alpha)}.$$

It is now natural to ask the question as to whether there exists
an operator $X$ closed, densely defined, affiliated with $\M$
which reduces to $Z^{(\alpha)}$ on $\M_\alpha$. To begin with, we
assume that the projections $\{P_\alpha\}$ are mutually
orthogonal. In this case, putting $\HH_\alpha = P_\alpha \HH$, we
have
$$ \HH= \bigoplus_{\alpha \in {\mc I}} \HH_\alpha =\{ (f_\alpha): \, f_\alpha \in \HH_\alpha, \; \sum_{\alpha
\in I}  \|f_\alpha\|^2 <\infty\}.$$ We put
$$ D(X)= \{ (f_\alpha) \in \HH: \, f_\alpha \in D(Z^{(\alpha)});\; \sum_{\alpha
\in I}  \|Z^{(\alpha)}f_\alpha\|^2 <\infty\}$$ and for
$f=(f_\alpha) \in D(X)$ we define
$$Xf= (Z^{(\alpha)}f_\alpha) .$$
Then
\begin{itemize}
  \item[(i)] $D(X)$ is dense in $\HH$.\\
  Indeed, $D(X)$ contains all $f=(f_\alpha)$ with $f_\alpha=0$
  except that for a finite subset of indeces.
  \item[(ii)]$X$ is closed in $\HH$.\\
  Indeed, let $f_n= (f_{n,\alpha})$ be a sequence of elements of
  $D(X)$ with $f_n \to g=(g_\alpha)\in \HH$ and $Xf_n \to h$.
  Since
  $$ f_n \to g \Leftrightarrow f_{n, \alpha} \to g_\alpha \in \HH_\alpha,
  \; \forall \alpha \in {\mc I}$$ and
$$ Xf_n \to h \Leftrightarrow (Xf_{n})_\alpha \to h_\alpha \in \HH_\alpha,
  \; \forall \alpha \in {\mc I},$$
  by $(Xf_{n})_\alpha=Z^{(\alpha)}f_{n, \alpha}$ and from the
  closedness of each $Z^{(\alpha)}$ in $\HH_\alpha$, we get
  $$ g_\alpha \in D(Z^{(\alpha)}) \hspace{4mm} \mbox{and}
  \hspace{4mm} h_\alpha = Z^{(\alpha)}g_\alpha. $$
  It remains to check that $\sum_{\alpha
\in {\mc I}}  \|Z^{(\alpha)}g_\alpha\|^2 <\infty$ but this is clear,
since both $(Z^{(\alpha)}g_\alpha)$ and $h=(h_\alpha) \in \HH$.
  \item[(iii)] $X \aff \M$.\\
  Let $Y\in \M'$. Then, $\forall f \in \HH$, $Yf=(YP_\alpha f)$
  and $YP_\alpha \in (\M P_\alpha)'= \M'P_\alpha$. Therefore
 $$ XYf =((XY)P_\alpha f) =(YXP_\alpha f) =YXf .$$
\end{itemize}
In conclusion, $X$ is a measurable operator.

\smallskip
Thus, we have proved the following
\begin{prop}Let $\FF =\{\vp_\alpha\}_{\alpha\in\mathcal{I}}$ be a sufficient
family of normal, finite traces on the von Neumann algebra
$\mathfrak{M}$. Assume that Condition (P) is fulfilled and that
the $\vp_\alpha$'s have mutually orthogonal supports. Then
${\mathfrak{M}}_p$, $p\geq 1$, consists of measurable operators.
\end{prop}

\medskip
The analysis of the general case would really be simplified if,
from a given sufficient family $\FF$ of normal finite traces, one
could extract (or construct) a {\it sufficient} subfamily ${\mc
G}$ of traces with mutually orthogonal supports. Apart from quite
simple situations (for instance when $\FF$ is finite or
countable), we do not know if this is possible or not. There is
however a relevant case where this can be fairly easily done. This
occurs when $\FF$ is a convex and $w^*$-compact family of traces
on $\M$.
\begin{lemma} \label{LEMMA3}
Let $\FF$ be a convex $w^{\ast}$-compact family of normal, finite
traces on a von Neumann algebra $\mathfrak{M}$; assume that, for
each central operator $Z$, with $0\leq Z\leq {\mb I}$, and each
$\eta\in\FF$ the functional $\eta_{\scriptscriptstyle
Z}(X):=\eta(XZ)$ belongs to $\FF$. Let $\mathfrak{E}\FF$ be the
set of extreme elements of $\FF$. If $\eta_1,
\eta_2\in\mathfrak{E}\FF$, $\eta_1 \neq n_2$, and $P_1$ and $P_2$
are their respective supports, then $P_1$ and $P_2$ are
orthogonal.
\end{lemma}

\begin{proof}
Let $P_1, P_2$ be, respectively, the supports of $\eta_1$ and $\eta_2$. We begin with proving that either $P_1=P_2$ or $P_1 P_2=0$. Indeed, assume that
 $P_1P_2\neq 0$. We define
$$\eta_{1,2}(X)=\eta_1(XP_2) \quad \quad X \in \M.$$
Were $\eta_{1,2}=0$, then, in particular $\eta_{1,2}(P_2)=0$, i.e.
$\eta_{1}(P_2)=0$ and therefore, by definition of support,
$P_2\leq 1-P_1$. This implies that $P_1P_2=0$, which contradicts
the assumption. We now show that the support of $\eta_{1,2}$ is
$P_1P_2.$ Let, in fact, $Q$ be a projection such that
$\eta_{1,2}(Q)=0$. Then
$$\eta_{1}(QP_2)=0 \,\Rightarrow QP_2 \leq 1-P_1\,\Rightarrow QP_2 (1-P_1) = QP_2 \Rightarrow QP_2P_1=0.$$
Then the largest $Q$ for which this happens is $ 1-P_2P_1.$ We
conclude that the support of the trace $\eta_{1,2}$ is $ P_1P_2.$
Finally, by definition, one has $\eta_{1,2}(X)=\eta_1(XP_2),$ and,
since $XP_2 \leq X$,
$$\eta_{1,2}(X)=\eta_1(XP_2)\leq \eta_1(X) \quad \forall X\in\mathfrak{M}.$$
Thus $\eta_{1}$ majorizes  $\eta_{1,2}.$ But $\eta_{1}$ is extreme
in $\FF.$ Therefore $\eta_{1,2}$ has the form $\lambda\eta_{1}$
with $\lambda\in ]0,1]$. This implies that $\eta_{1,2}$ has the
same support as $\eta_{1}$; therefore $ P_1P_2=P_1$ i.e. $ P_1
\leq P_2.$ Starting from $\eta_{2,1}(X)=\eta_2(XP_1)$, we get, in
similar way, $ P_2 \leq P_1.$ Therefore, $P_1P_2 \neq 0$ implies
$P_1=P_2$. However,  two different traces of $\mathfrak{E}\FF$
cannot have the same support. Indeed, assume that there exist
$\eta_1, \eta_2\in\FF$ having the same support $P.$ Since $P$ is
central, we can consider the von Neumann algebra $\mathfrak{M}P$.
The restrictions of $\eta_1, \eta_2$ to $\mathfrak{M}P$ are normal
faithful semifinite traces. By \cite[Prop. V.2.31]{takesaki} there
exist a central element $Z$ in $\mathfrak{M}P$ with $0\leq  Z \leq
P$ ($P$ is here considered as the unit of $\M P$) such that
\begin{equation} \label{uno}\eta_1(X)=(\eta_1+\eta_2)(ZX) \quad \forall
X\in\mathfrak({\M}P)_+.\end{equation} Then $Z$ also belongs to the
center of $\mathfrak{M}$, since for every $V\in\mathfrak{M}$
$$ZV=Z(VP+VP^\perp)=ZVP=VZP=VZ.$$
Therefore the functionals
$$\eta_{\scriptscriptstyle {1,Z}}(X):=\eta_1(XZ) \quad \quad \eta_{\scriptscriptstyle {2,Z}}(X):=\eta_2(XZ) \quad\quad  X\in\mathfrak{M}$$
belong to the family $\FF$ and are majorized, respectively, by the
extreme elements $\eta_1,\eta_2.$ Then, there exist $\lambda,\mu
\in[0,1]$ such that
$$\eta_1(XZ)=\lambda\eta_1(X)\quad \quad \eta_2(XZ)=\mu\eta_1(X), \quad \forall X \in \M.$$
If $\lambda=1$ we would have, from \eqref{uno}, $\eta_2(ZX)=0$,
for every $X\in(\mathfrak{M}P)_+$; in particular, $\eta_2(\mid Z
\mid^2)=0$; this implies that $ Z=0$. Thus $\lambda \neq 1.$
Analogously, $\mu \neq 0$; indeed, if $\mu=0$, then
$\eta_1(X)=\lambda\eta_1(X)$ and thus $\lambda=1.$ Therefore there
exist $\lambda,\mu\in(0,1)$ such that
$$\eta_1(X)=\lambda\eta_1(X)+\mu\eta_2(X) \quad \forall
X\in\mathfrak{M}P ,$$ which, in turn, implies
$$\eta_1(X)=\lambda\eta_1(X)+\mu\eta_2(X) \quad \forall
X\in\mathfrak{M}$$ Hence,
$$(1-\lambda)\eta_1(X)=\mu\eta_2(X) \quad \forall
X\in\mathfrak{M}.$$ From the last equality, dividing by
$\max\{1-\lambda,\mu\}$ one gets that one of the two elements is a
convex combination of the other and of $0$; which is absurd. In
conclusion, different supports of extreme traces of ${\mathfrak
F}$ are orthogonal.

\end{proof}

Since, for every $X \in \M$, $\|X\|_{p, {\mc I}}$ remains the same
if computed either with respect to ${\mathfrak F}$ or to
${\mathfrak{EF}}$, we can deduce the following

\begin{thm}\label{casoA}
Let $\FF$ be a convex and $w^*$-compact sufficient family of
normal, finite traces on the von Neumann algebra $\mathfrak{M}$.
Assume that $\FF$ satisfies Condition (P) and that for each
central operator $Z$, with $0\leq Z\leq {\mb I}$, and each
$\eta\in\FF$ the functional $\eta_{\scriptscriptstyle
Z}(X):=\eta(XZ)$ belongs to $\FF$. Then the completion
${\mathfrak{M}}_p[\norm{\cdot}_{p,\mathcal{I}}], $ consists of
measurable operators.
\end{thm}
Families of traces satisfying the assumptions of Theorem
\ref{casoA} will be constructed in the next section.
\section{A representation theorem}
Once we have constructed in the previous section some CQ*-algebras
of operators affiliated to a given von Neumann algebra, it is
natural to pose the question under which conditions can an
abstract CQ*-algebra $(\X,\Ao)$ be realized as a CQ*-algebra of
this type.

 Let $(\X[\norm{\cdot}], \Ao
[\norm{\cdot}_{\scriptscriptstyle 0}])$ be a CQ*-algebra with unit
$e$ and let
$$ \mathcal{T}(\X)= \{\Omega \in {\mc S}(\X):\, \Omega(x,x)=\Omega(x^*,x^*), \;\forall x \in \X\}.$$
We remark that if $\Omega \in {\mc T}(\X)$ then, by polarization,
$\Omega(y^ {\ast},x^ {\ast})=\Omega(x,y),\quad \forall x,y \in
\X.$\\ It is easy to prove that the set $ \mathcal{T}(\X)$ is
convex.

\medskip
For each $\Omega\in\mathcal{T}(\X)$, we define a linear
functional $\omega_\Omega$ on $\Ao$ by
$$\omega_\Omega(a):=\Omega(a,e)\quad \quad a\in\Ao.$$
We have
$$\omega_\Omega(a^{\ast}a)=\Omega(a^{\ast}a,e)=\Omega(a,a)=\Omega(a^{\ast},a^{\ast})=\omega_\Omega(aa^{\ast})\geq0.$$
This shows at once that $\omega_\Omega$ is positive and tracial.\\
We put
$$\mathfrak{M}_\mathcal{T}(\Ao)=\{\omega_\Omega;\, \Omega\in\mathcal{T}(\X)\}.$$
>From the  convexity of $\mathcal{T}(\X)$ it follows easily that
$\mathfrak{M}_\mathcal{T}(\Ao)$ is convex too. If we denote with
$\norm{f}^\sharp$ the norm of the bounded functional $f$ on $\Ao$,
we also get
$$\norm{\omega_{\Omega}}^\sharp=\omega_{\Omega}(e)=\Omega(e,e)\leq\norm{e}^2.$$
Therefore
$$\mathfrak{M}_\mathcal{T}(\Ao)\subseteq\{\omega\in\Ao^\sharp:
\,\norm{\omega}^\sharp\leq\norm{e}^2\},$$ where $\Ao^\sharp$
denotes the topological dual of $\Ao
[\norm{\cdot}_{\scriptscriptstyle 0}].$\\
Setting
$$f_\Omega(a):=\tfrac{\omega_\Omega(a)}{\norm{{e}}^2}$$
we get $$ f_\Omega\in\{\omega\in\Ao^\sharp: \norm{\omega}^\sharp
\leq 1 \}.$$ By the Banach - Alaglou theorem, the set
$\{\omega\in\Ao^{\sharp}: \norm{\omega}^{\sharp} \leq 1 \} $ is a
$w^{\ast}-$ compact subset of $\Ao^{\sharp}$ . Then, the set
$\{\omega\in\Ao^\sharp: \,\norm{\omega}^{\sharp}\leq\norm{e}^2\}$
is also $w^{\ast}-$compact.

\begin{prop} \label{compatto}
$\mathfrak{M}_\mathcal{T}(\Ao)$ is $w^{\ast}$-closed and,
therefore, $w^{\ast}$-compact.
\end{prop}

\begin{proof}
Let $(\omega_{\Omega_\alpha})$ be a  net in
$\mathfrak{M}_\mathcal{T}(\Ao)$ $w^{\ast}-$ converging to a
functional $\omega\in\Ao^{\sharp}.$ We will show that
$\omega=\omega_{\Omega}$ for some
$\Omega\in\mathcal{T}(\X).$ Let us begin with defining
$\Omega_{\scriptscriptstyle 0}(a,b)=\omega(b^{\ast}a),\,a,b\in\Ao.$ By the definition
itself, $(\omega_{\Omega_\alpha})(a)\longrightarrow
\omega(a)=\Omega_{\scriptscriptstyle 0}(a,e)$. Moreover, for every
$a,b\in\Ao$,
$$ \Omega_{\scriptscriptstyle 0}(a,b)=\omega(b^{\ast}a)=\lim_{\alpha}\omega_{\Omega_\alpha}(b^{\ast}a)=\lim_{\alpha}\Omega_\alpha(a,b).$$
Therefore $$ \Omega_{\scriptscriptstyle
0}(a,a)=\lim_{\alpha}\Omega_\alpha(a,a)\geq 0.$$ We also have
$$\mid\Omega_{\scriptscriptstyle 0}(a,b)\mid=\lim_{\alpha}\mid\Omega_\alpha(a,b)\mid
\leq \norm{a}\, \norm{b}.$$ Hence $\Omega_{\scriptscriptstyle 0}$
can be extended by continuity to $\X\times\X.$ \ Indeed, let
$$x=\| \cdot \|-\lim_{n}a_n \quad y=\| \cdot \|-\lim_{n}b_n \quad (a_n),
(b_n)\subseteq\Ao$$ then
$$\mid\Omega_{\scriptscriptstyle 0}(a_n,b_n)-\Omega_{\scriptscriptstyle 0}(a_m,b_m)\mid=\mid
\Omega_{\scriptscriptstyle 0}(a_n,b_n)-\Omega_{\scriptscriptstyle 0}(a_m,b_n)+\Omega_{\scriptscriptstyle 0}(a_m,b_n)-\Omega_{\scriptscriptstyle 0}(a_m,b_m)
\mid\leq$$
$$\leq \mid\Omega_{\scriptscriptstyle 0}(a_n-a_m,b_n)\mid+
\mid\Omega_{\scriptscriptstyle 0}(a_m,b_n-b_m)\mid \leq
\norm{a_n-a_m} \norm{b_n}\,+\, \norm{a_m}\norm{b_n-b_m}\rightarrow
0 , $$ since $(\norm{a_n})$ ed $(\norm{b_n})$ are bounded
sequences. Therefore we can define
$$\Omega(x,y)=\lim_{n}\Omega_{\scriptscriptstyle 0}(a_n,b_n).$$
Clearly, $\Omega(x,x)\geq 0\quad \forall x\in\X.$ \\ It is easily
checked that $\Omega\in\mathcal{T}(\X)$. This concludes the proof.
\end{proof}

Since $\mathfrak{M}_\mathcal{T}(\Ao)$ is convex and
$w^{\ast}$-compact, by the Krein-Milmann theorem it follows that
it has extreme points and it coincides with the $w^{\ast}$-closure
of the convex hull of the set $\mathfrak{EM}_\mathcal{T}(\Ao)$ of
its extreme points.

By the Gelfand - Naimark theorem each $C^{\ast}$-algebra is
isometrically *-isomorphic to a $C^{\ast}$-algebra of bounded
operators in Hilbert space. This isometric *-isomorphism is called
the {\it universal *-representation}.

Thus, let $\pi$ be the universal *-representation of $\Ao$ and
$\pi(\Ao)^{''}$ the von Neumann algebra generated by
$\pi(\Ao).$

For every $\Omega\in\mathcal{T}(\X)$ and $a\in\Ao$, we put
$$\vp_\Omega(\pi(a))=\omega_\Omega(a).$$

Then, for each $\Omega\in \mathcal{T}(\X)$, $\vp_\Omega$
is a positive bounded linear functional on the operator algebra
$\pi ( \Ao).$

Clearly,
$$\vp_\Omega(\pi(a))=\omega_\Omega(a)=\Omega(a,e)$$

$$\mid\vp_\Omega(\pi(a))\mid=\mid\omega_\Omega(a)\mid=
\mid\Omega(a,e)\mid\leq \| a \| \| e\| \leq  \| a \|_{\scriptscriptstyle 0} \|e\|^2=\|
\pi(a) \| \, \| e\|^2 .$$

Thus $\vp_\Omega$ is continuous on $\pi(\Ao).$

By \cite[Theorem 10.1.2]{kadison}, $\vp_\Omega$ is weakly
continuous and so it extends uniquely to $\pi(\Ao)^{''}$.
Moreover, since $\vp_\Omega$ is a trace on $\pi(\Ao)$, the
extension  $\widetilde{\vp_\Omega}$ is a trace on
$\mathfrak{M}:=\pi(\Ao)^{''}$ too.

The norm $\norm{\widetilde{\vp_\Omega}}^\sharp$ of
$\widetilde{\vp_\Omega}$ as a linear functional on $\mathfrak{M}$
equals the norm of $\vp_\Omega$ as a functional on
$\pi(\Ao).$

We have:
$$\norm{\widetilde{\vp_\Omega}}^\sharp={\widetilde\vp_\Omega}(\pi(e)) =\vp_\Omega(\pi(e))=\omega_\Omega(e)\leq
\|e\|^2.
$$

The set
$$\mathfrak{N}_\mathcal{T}(\Ao)=\{\widetilde{\vp_\Omega};\,
\Omega\in\mathcal{T}(\X) \}$$ is convex and
$w^{\ast}$-compact in $\mathfrak{M}^\sharp$, as can be easily
seen by considering the map
$$\omega_\Omega\in\mathfrak{M}_\mathcal{T}(\Ao)\rightarrow
\widetilde{\vp_\Omega}\in\mathfrak{N}_\mathcal{T}(\Ao)$$ which is
linear and injective and taking into account the fact that, if
$a_\alpha\rightarrow a$ in $\Ao
[\norm{\cdot}]$, then $\widetilde{\vp_\Omega}(\pi(a_\alpha)-\pi(a))=
\omega_\Omega(a_\alpha-a)\rightarrow 0.$

Let $\mathfrak{EN}_\mathcal{T}(\Ao)$ be the set of extreme points
of $\mathfrak{N}_\mathcal{T}(\Ao)$; then
$\mathfrak{N}_\mathcal{T}(\Ao)$ coincides with $w^{\ast}$-closure
of the convex hull of $\mathfrak{EN}_\mathcal{T}(\Ao).$ The
extreme elements of $\mathfrak{N}_\mathcal{T}(\Ao)$ are easily
characterized by the following

\begin{prop}
$\widetilde{\vp_\Omega}$ is extreme in
$\mathfrak{N}_\mathcal{T}(\Ao)$ if, and only if, $\omega_\Omega$
is extreme in $\mathfrak{M}_\mathcal{T}(\Ao).$
\end{prop}

\begin{defn}
A Banach quasi *-algebra  $(\X[\norm{\cdot}], \Ao
[\norm{\cdot}_{\scriptscriptstyle 0}])$ is said to be strongly
regular if ${\mc T}(\X)$ is sufficient and
$$ \|x\|= \sup_{\Omega \in {\mc T}(\X)} \Omega(x,x)^{1/2}, \quad \forall x \in \X.$$
\end{defn}

 \beex
If $\M$ is a von Neumann algebra possessing a sufficient family
$\FF$ of normal finite traces, then the CQ*-algebra $({\M}_p,\M)$
constructed in Section 3 is strongly regular. This follows from
the definition itself of the norm in the completion. \enex
 \beex If
$\vp$ is a normal faithful finite trace on $\M$, then ${\mc
T}(L^p(\vp))$, for $p\geq 2$, is sufficient. To see this, we start
with defining $\Omega_{\scriptscriptstyle 0}$ on $\M \times \M$ by
$$ \Omega_{\scriptscriptstyle 0} (X,Y)= \vp(Y^*X), \quad X,Y \in \M.$$
Then
$$|\Omega_{\scriptscriptstyle 0} (X,Y)|= |\vp(Y^*X)| \leq \|X\|_p \|Y\|_{p'}, \quad \forall X,Y \in \M.$$
Since $p\geq 2$, then $L^p(\vp)$ is continuously embedded into
$L^{p'}(\vp)$. Thus, there exists $\gamma >0$ such that
$\|Y\|_{p'} \leq \gamma \|Y\|_p$ for every $Y \in \M$. Let us
define
$$ \widetilde{\Omega}(X,Y)= \frac{1}{\gamma} \Omega_{\scriptscriptstyle 0} (X,Y), \quad \forall X,Y\in \M.$$
Then
$$ |\widetilde{\Omega}(X,Y)| \leq \|X\|_p \|Y\|_{p}, \quad \forall X,Y \in \M.$$
Hence, $\widetilde{\Omega}$ has a unique extension, denoted with
the same symbol, to $L^p(\vp) \times L^p(\vp)$. It is easily seen
that $\widetilde{\Omega} \in {\mc T}(L^p(\vp))$.

Were, for some $X \in L^p(\vp)$, $\Omega(X,X)=0$, for every
$\Omega  \in {\mc T}(L^p(\vp))$, we would then have
$\widetilde{\Omega}(X,X)=\|X\|^2_2=0$. This, clearly, implies
$X=0$. The equality $\widetilde{\Omega}(X,X)=\|X\|^2_2$ also shows
that $L^2(\vp)$ is strongly regular.
 \enex
 Let now
$(\X[\norm{\cdot}], \Ao [\norm{\cdot}_{\scriptscriptstyle 0}])$ be
a CQ*-algebra with unit $e$ and sufficient ${\mc T}(\X)$. Let
 $\pi:\Ao \hookrightarrow \BH$ be the universal representation of $\Ao$. Assume that the $C^*$algebra
$\pi(\Ao):=\mathfrak{M}$ is a von Neumann algebra. In this case,
$\mathfrak{M}_\mathcal{T}(\Ao)=\mathfrak{N}_\mathcal{T}(\Ao)$ and
$\mathfrak{N}_\mathcal{T}(\Ao)$ is a family of traces satisfying
Condition (P). Therefore, by Proposition \ref{nota}, we can
construct for $p\geq 1$, the CQ*-algebras
$({\mathfrak{M}}_p[\norm{\cdot}_{p,\,{\scriptscriptstyle\mathfrak{N}_\mathcal{T}(\Ao)}}],
\mathfrak{M}[\norm{\cdot}])$. Clearly, $\Ao$ can be identified
with $\M$. It is then natural to pose the question if also $\X$
can be identified with some ${\mathfrak{M}}_p$. The next Theorem
provides the answer to this question.

\begin{thm}
Let $(\X[\norm{\cdot}], \Ao [\norm{\cdot}_{\scriptscriptstyle
0}])$ be a CQ*-algebra with unit $e$ and and sufficient ${\mc
T}(\X)$.

 Then there
exist a von Neumann algebra $\M$ and a monomorphism $$\Phi: x\in\X
\rightarrow \Phi(x):=\widetilde{X}\in\M_2$$ with the following
properties:

(i) $\Phi$ extends the universal *-representation $\pi$ of $\Ao$;

 (ii) $\Phi(x^*)=\Phi(x)^*,
\quad \forall x \in \X$;

 (iii) $\Phi(xy)=\Phi(x)\Phi(y)$ for every
$x,y\in \X$ such that $x \in \Ao$ or $y \in \Ao$. \\Then $\X$ can
be identified with a space of operators affiliated with $\M$.

If, in addition, $(\X, \Ao)$ is strongly regular, then

(iv) $\Phi$ is an isometry of $\X$ into $\M_2$;

(v) If $\Ao$ is a W*-algebra, then $\Phi$ is an isometric
*-isomorphism of $\X$ onto
  $\M_2$.
\noindent
\end{thm}
\begin{proof} Let $\pi$ be the universal representation of $\Ao$ and
assume first that $\pi(\Ao)=:\mathfrak{M}$ is a von Neumann
algebra. By Proposition \ref{compatto}, the family of traces
$\mathfrak{M}_\mathcal{T}(\Ao)$ is convex and $w^{\ast}$-compact.
Moreover,  for each central positive element $Z$ with $0\leq Z
\leq \id \,$ and for $\vp\in\mathfrak{M}_\mathcal{T}(\Ao)$, the
trace $\vp_{\scriptscriptstyle Z}(X):=\vp(ZX)$ yet belongs to
$\mathfrak{M}_\mathcal{T}(\Ao).$ Indeed, starting from the form
$\Omega\in\mathcal{T}(\X)$ which generates $\vp$, one can define
the sesquilinear form
$$\Omega_{\scriptscriptstyle Z}(x,y):=\Omega(x \pi^{-1}(Z^{1/2}), y \pi^{-1}(Z^{1/2}))
 \quad \forall x,y \in \X.$$
We check that $\Omega_{\scriptscriptstyle Z}\in\mathcal{T}(\X).$

(i) $\Omega_{\scriptscriptstyle Z}(x,x)=\Omega(x
\pi^{-1}(Z^{1/2}), x \pi^{-1}(Z^{1/2}))\geq 0, \quad \forall
x\in\X$

(ii) We have, for every $x \in \X$ and for every $ a,b\in\Ao$,
\begin{eqnarray*}\Omega_{\scriptscriptstyle Z}(xa,b)&=&\Omega(xa \pi^{-1}(Z^{1/2}), b
\pi^{-1}(Z^{1/2}))\\ &=& \Omega(a
\pi^{-1}(Z^{1/2}), x^ {\ast}b \pi^{-1}(Z^{1/2})) \\
&=&\Omega_{\scriptscriptstyle Z}(a,x^ {\ast}b).\end{eqnarray*}

(iii) We have, for every $x,y \in \X$,
\begin{eqnarray*}\mid\Omega_{\scriptscriptstyle Z}(x,y)\mid &=& \mid\Omega(x \pi^{-1}(Z^{1/2}), y
\pi^{-1}(Z^{1/2}))|\\ &\leq& \norm{x\pi^{-1}(Z^{1/2})}\,
\norm{\pi^{-1}(Z^{1/2})y}\\ &\leq & \norm{x}
\norm{\pi^{-1}(Z^{1/2})}_{\scriptscriptstyle 0} \, \norm{y}
\norm{\pi^{-1}(Z^{1/2})}_{\scriptscriptstyle 0}\\ &\leq &
\norm{x}\, \norm{y}.\end{eqnarray*}

(iv) For every $x \in\X$,
\begin{eqnarray*} \Omega_{\scriptscriptstyle Z}(x^ {\ast},x^
{\ast})&=&\Omega(x^ {\ast} \pi^{-1}(Z^{1/2}), x^ {\ast}
\pi^{-1}(Z^{1/2}))\\ &=& \Omega(x \pi^{-1}(Z^{1/2}), x
\pi^{-1}(Z^{1/2})) = \Omega_{\scriptscriptstyle Z}(x,x).\end{eqnarray*}

Moreover, $\Omega_{\scriptscriptstyle Z}$ defines, for every
$A=\pi(a)\in\mathfrak{M}=\pi(\Ao)$, the following trace
\begin{eqnarray*} \vp_{\Omega_{\scriptscriptstyle Z}}(A)&=&\Omega_{\scriptscriptstyle Z}(a,e)=\Omega(a \pi^{-1}(Z^{1/2}),
\pi^{-1}(Z^{1/2}))\\ &=& \Omega(a \pi^{-1}(Z), e)=\Omega(
\pi^{-1}(AZ), e)=\vp_{\Omega}(AZ)   \end{eqnarray*} Then, the
family of traces
$\mathfrak{N}_\mathcal{T}(\Ao)\,(=\mathfrak{M}_\mathcal{T}(\Ao))$
satisfies the assumptions of Lemma \ref{LEMMA3}; therefore, if
$\eta_1, \eta_2\in\mathfrak{EN}_\mathcal{T}(\Ao)$, denoting with
$P_1$ and $P_2$ their respective supports, one has $P_1P_2=0.$

By the sufficiency of ${\mc T}(\X)$ we get
$$\|X\|_{2,{\scriptscriptstyle\mathfrak{M}_\mathcal{T}(\Ao)}}:=\sup_{\vp\in
\mathfrak{M}_\mathcal{T}(\Ao)}\|X\|_{2,\vp}=\sup_{\vp\in
\mathfrak{EM}_\mathcal{T}(\Ao)}\|X\|_{2,\vp} \quad \forall
X\in{\pi(\Ao)}.$$ By Proposition \ref{nota}, the Banach space
$\M_2$, completion of $\mathfrak{M}$ with respect to the norm $\|
\cdot
\|_{{2,\,{\scriptscriptstyle\mathfrak{N}_\mathcal{T}(\Ao)}}}$ , is
a CQ*-algebra. Moreover, since the supports of the extreme traces
satisfy the assumptions of Theorem \ref{casoA}, the CQ*-algebra
$(\M_2[\norm{\cdot}_{2,{\scriptscriptstyle\mathfrak{N}_\mathcal{T}(\Ao)}}],
\mathfrak{M}[\norm{\cdot}])$, consists of operators affiliated
with $\mathfrak{M}$. We now define the map $\Phi$.

For every element $x\in\X$, there exists a sequence $\{a_n\}$ of
elements of $\Ao$ converging to $x$ with respect to the norm of
$\X (\| \cdot \|)$. Put $X_n= \pi(a_n)$, $n \in {\mb N}$. Then,
\begin{eqnarray*}\|X_n-X_m\|_{2,{\scriptscriptstyle\mathfrak{N}_\mathcal{T}(\Ao)}}
&:=&\sup_{\vp\in\mathfrak{N}_\mathcal{T}(\Ao)}\|\pi(a_n)-\pi(a_m)\|_{2,\vp}\\
&=&\sup_{\Omega\in \mathcal{T}(\X)}
[\Omega((a_n-a_m)^*(a_n-a_m),e)]^{1/2} \\ &=& \sup_{\Omega\in
\mathcal{T}(\X)} [\Omega(a_n-a_m,a_n-a_m)]^{1/2}\leq \|a_n-a_m\|
\rightarrow 0.\end{eqnarray*} Let $\widetilde{X}$ be the
$\|\cdot\|_{2,{\scriptscriptstyle\mathfrak{M}_\mathcal{T}(\Ao)}}$-limit
of the sequence $({X_n})$ in ${\mathfrak{M}}_2$.
We define $\Phi(x):=\widetilde{X}.$

For each $x \in \X$, we put $$p_{\scriptscriptstyle {\mc
T}(\X)}(x)= \sup_{\Omega \in {\mc T}(\X)} \Omega(x,x)^{1/2}.$$
Then, owed to the sufficiency of ${\mc T}(\X)$,
$p_{\scriptscriptstyle {\mc T}(\X)}$ is a norm on $\X$ weaker than
$\| \cdot \|$. This implies that
$$\|\widetilde{X}\|^2_{2,{\scriptscriptstyle\mathfrak{N}_\mathcal{T}(\Ao)}}
= \lim_{n \to \infty}\sup_{\Omega\in
\mathcal{T}(\X)}\Omega(a_n,a_n)= \lim_{n \to
\infty}p_{\scriptscriptstyle {\mc T}(\X)}(a_n)^2 =
p_{\scriptscriptstyle {\mc T}(\X)}(x)^2.$$

>From this equality it follows easily that the linear map $\Phi$ is
well defined and injective. The condition (iii) can be easily
proved. If $(\X, \Ao)$ is strongly regular, then,  for every $x
\in \X$, $ p_{\scriptscriptstyle {\mc T}(\X)}(x)= \|x\|$.  Thus
$\Phi$ is isometric. Moreover, in this case, $\Phi$ is surjective;
indeed, if ${T}\in{\mathfrak{M}_2}$, then there exists a sequence
${T_n}$ of bounded operators of $\pi(\Ao)$ which converges to
${T}$ with respect to the norm
$\|\cdot\|_{2,{\scriptscriptstyle\mathfrak{N}_\mathcal{T}(\Ao)}}$.
The corresponding sequence $\{t_n\}\subset\Ao$, $T_n= \Phi(t_n)$,
converges to $t$ with respect to the norm of $\X$ and $\Phi(t)=T$
by definition. Therefore $\Phi$ is an isometric *-isomorphism.

To complete the proof, it is enough to prove that the given
CQ*-algebra $(\X,\Ao)$ can be embedded in  a CQ*-algebra
$(\GK,\Bo)$ where $\Bo$ is a W*-algebra. Of course, we may
directly work with $\pi(\Ao)$ with $\pi$ the universal
representation of $\Ao$. The family of traces
$\mathfrak{N}_\mathcal{T}(\Ao)$ defined on $\pi(\Ao)''$ is not
necessarily sufficient. Let $P_\Omega$, $\Omega \in {\mc T}(\X)$,
denote the support of $\widetilde{\vp}_\Omega$ and let
$$P=\bigvee_{\Omega\in\mathcal{T}(\X)}{P_\Omega}.$$
Then $\Bo:=\pi(\Ao)''P$ is a von Neumann algebra, that we can
complete with respect to the norm
$$ \|X\|
_{2,{\scriptscriptstyle\mathfrak{N}_\mathcal{T}(\Ao)}}=
\sup_{\Omega \in {\mc T}(\X)}\widetilde{\vp}_\Omega(X^*X), \quad X
\in \pi(\Ao)''P.$$ We obtain in this way a CQ*-algebra $(\GK,\Bo)$
with $\Bo$  a W*-algebra. The faithfullness of $\pi$ on $\Ao$
implies that
$$\pi(a)P=\pi(a), \quad \forall a \in \Ao.$$
It remains to prove that $\X$ can be identified with a subspace of
$\GK$. But this can be shown in the very same way as we did in the
first part: for each $x\in \X$ there exists a sequence
$\{a_n\}\subset \Ao$ such that $\|x-a_n\|\to 0$ as $n \to \infty$.
We now put $X_n=\pi(a_n)$. Then, proceeding as before, we
determine the element $\hat{X}\in \GK$, where
$$\hat{X}=\| \cdot \|
_{2,{\scriptscriptstyle\mathfrak{N}_\mathcal{T}(\Ao)}}-\lim
\pi(a_n)P.$$ It is easy to see that the map $x \in \X \to
\hat{X}\in \GK$ is injective. If $(\X,\Ao)$ is regular, but
$\pi(\Ao)\subset \pi(\Ao)''$, then $\Phi$ is an isometry of $\X$
into $\M_2$, but needs not be surjective.
\end{proof}

\vspace{6mm} {\noindent}{\bf Acknowledgment} We acknowledge
financial support of MIUR through national and local grants.

\bibliographystyle{amsplain}

\end{document}